\documentclass[a4paper, 12pt]{amsart}
 \usepackage[french, english]{babel}
 \usepackage[normalem]{ ulem }
\usepackage{soul}

\usepackage[T1]{fontenc}
\usepackage{amsmath,amsfonts,amssymb,amsopn,amscd,amsthm}
\usepackage{tikz-cd}
\usepackage{gensymb}
\usepackage{comment}
\usepackage{dsfont}
\usepackage{graphicx}
\usepackage{color}
\usepackage[colorlinks]{hyperref}
\usepackage{epigraph}
\usepackage{enumerate}

\usepackage{marginnote}
\usepackage{todonotes}
\usepackage[left=3.0cm,top=3.0cm,bottom=3.0cm,right=3.0cm, marginparwidth=3cm]{geometry}

\title[ Toda and G\texorpdfstring{$\beta$}{B}E ]
      { Scattering of the Toda system and the Gaussian \texorpdfstring{$\beta$}{B}-ensemble
      }
      
\author{Reda   \textsc{Chhaibi}}
\address{Institut de math\'ematiques de Toulouse, UMR5219, Universit\'e de Toulouse,  118 route de Narbonne, F-31062 Toulouse Cedex 9, France}
\email{reda.chhaibi@math.univ-toulouse.fr}

\date{\today}

\allowdisplaybreaks[4]

\DeclareMathOperator{\id}{id}
\DeclareMathOperator{\Tr}{Tr}

\def\half{\frac{1}{2}}

\def\1{{\mathbf 1}}

\def\C{{\mathbb C}}

\def\R{{\mathbb R}}

\def\P{{\mathbb P}}
\def\E{{\mathbb E}}

\def\Lc{{\mathcal L}}

\def\Nc{{\mathcal N}}

\def\Sc{{\mathcal S}}
\def\Tc{{\mathcal T}}

\def\afrak{{\mathfrak a}}

\def\kfrak{{\mathfrak k}}
\def\nfrak{{\mathfrak n}}

\numberwithin{equation}{section}
\numberwithin{figure}{section}

\newtheorem{thm}{Theorem}[section]
\newtheorem{proposition}[thm]{Proposition}

\newtheorem{rmk}[thm]{Remark}

\begin{document}

\begin{abstract}
   The classical Toda flow is a well-known integrable Hamiltonian system that diagonalizes matrices. By keeping track of the distribution of entries and precise scattering asymptotics, one can exhibit matrix models for log-gases on the real line. These types of scattering asymptotics date back to fundamental work of Moser.

   More precisely, using the classical Toda flow acting on symmetric real tridiagonal matrices, we give a "symplectic" proof of the fact that the Dumitriu-Edelman tridiagonal model has a spectrum following the Gaussian $\beta$-ensemble. 
\end{abstract}

\maketitle

\setcounter{tocdepth}{4}
\medskip
\medskip
\medskip
\hrule
\tableofcontents
\hrule

\section{Gaussian \texorpdfstring{$\beta$}{B}-ensembles and the Macdonald-Mehta-Opdam formula}
The Gaussian $\beta$-ensemble is the probability distribution for an $n$-point configurations in the real line:
\begin{align}
\label{def:beta_ensemble}
\left( G\beta E_n \right) \quad \P\left( \Lambda \in dx \right) & := \frac{1}{Z^\beta_n} \left| \Delta(x) \right|^\beta e^{-\half \sum_{j=1}^n x_i^2 } \prod_{j=1}^n dx_i \ ,
\end{align}
with $Z_n^\beta$ being the normalization constant. It is a particular case of $\beta$-ensembles with general confining potential $V$:
\begin{align}
   \label{def:general_beta_ensemble}
   \P\left( \Lambda \in dx \right) & := \frac{1}{Z^\beta_n} \left| \Delta(x) \right|^\beta e^{- \sum_{j=1}^n V(x_i) } \prod_{j=1}^n dx_i \ .
\end{align}

From \cite{bib:DumitriuEdelman}, the $G \beta E_n$is conveniently obtained as the spectrum of the tridiagonal matrix:
\begin{align}
\label{eq:T_beta}
   T_\beta = \begin{pmatrix}
               \Nc_1                & \chi_{\half (n-1) \beta} &     0                      & \dots & 0                        & 0\\
               \chi_{\half (n-1)\beta} & \Nc_2                 &   \chi_{\half (n-2) \beta} & \dots & 0                        & 0\\
               \dots                   & \dots                    &   \dots                    & \dots & \dots                    & \dots\\
               0                     &     0                    &     0                      & \dots & \Nc_{n-1}                 & \chi_{\half \beta}\\
               0                     &     0              &     0                            & \dots & \chi_{\half \beta}       & \Nc_n\\
               \end{pmatrix}
\end{align}
where the variables with different symbols are independent. $\chi_k$ stands for a $\chi$-distributed random variable with $k$ degrees of freedom and $\Nc_i$ stands for a standard Gaussian random variable. We record this fact for future reference as

\begin{thm}[Dumitriu and Edelman, \cite{bib:DumitriuEdelman}]
\label{thm:DumitriuEdelman}
The spectrum of $T_\beta$, as given in Eq. \eqref{eq:T_beta}, is the $\beta$-ensemble on the line, with quadratic confinement potential.

Moreover, the Macdonald-Mehta-Opdam identity holds true:
$$ Z^\beta(\R^n)   = \int_{\R^n} dx\ \left| \Delta(x) \right|^\beta e^{-\half \sum_{j=1}^n x_i^2 }
                   =
                     \left( 2 \pi \right)^{ \half n} 
                     \prod_{j=1}^n \frac{\Gamma\left( 1 + \frac{\beta}{2} j \right)}
                                        {\Gamma\left( 1 + \frac{\beta}{2}   \right)} \ .$$
\end{thm}

Thanks to this tridiagonal model, the spectra of $G\beta E_n$ and $G \beta E_{n+1}$ are naturally coupled. This is the same coupling in $\beta=2$ of Gaussian measure on infinite Hermitian matrices studied by Olshanki-Vershik \cite{olshanski1996ergodic}, after the so-called Trotter reduction.

From a physical perspective, the spectrum gives a log-gas with quadratic confinement potential. The fact that Dumitriu and Edelman's model has independent entries is miraculous and shows that $\beta$-ensembles are integrable in a sense.

The present work provides another proof that Dumitriu and Edelman's triadiagonal model has a spectrum distributed according to the $\beta$-ensemble on the line. While there is no new result per se, it is the approach that is novel. This derivation uses a Hamiltonian technique based on the scattering for the Toda flow. It is perhaps ``the symplectic proof'' that Dumitriu and Edelman mention in their paper \cite{bib:DumitriuEdelman} in the form of their Remark 2.10. Also, since this geometric proof splits the space into independent entries, it gives the change of variables which produces a proof of the Macdonald-Mehta-Opdam (MMO) integral.

The general approach goes as follows. We consider a random matrix with fixed distribution as starting point for the Toda flow. As the Toda flow is an integrable dynamical system that diagonalizes matrices, we are able to keep track of the matrix distribution throughout the flow.

Another interesting point regards the nature of the integrability of such models and we shed some light on the matter. Indeed, with initial measures expressed in terms of Casimirs, i.e. invariants of motion, the spectral distribution has a tractable and closed form expression. The particularity of quadratic potentials is that they are expressed using the first Casimir only, whose special structure leads to independence in the entries of the matrix model.

\medskip

{\bf Summary. }
We start by developing the necessary results from the theory of the Toda lattice in Section \ref{section:Toda_flow_definition}. 

Then we explicitly compute the scattering asymptotics, using orthogonal polynomials techniques. We will be particularly interested in precise scattering asymptotics which are originally due to Moser \cite{bib:Moser75} and which show the appearance of a logarithmic interaction via a Vandermonde. This allows to construct a scattering map between the generalised Toda flow and a free dynamic. 

Before proving the main result, we find invariant volume forms under the Toda flow. These will play the role of reference measures. Putting everything together in Section \ref{section:Toda_proof_main} shows that, indeed, the spectrum is distributed as \eqref{def:beta_ensemble}. The MMO formula is obtained by keeping track of the normalizing constant.

\section{Definition of the Toda flow}
\label{section:Toda_flow_definition}
\
\medskip

{\bf Notations.}
Let $\kfrak$, $\afrak$ and $\nfrak$ be respectively the subspaces in $M_n\left( \C \right)$ of anti-Hermitian matrices, diagonal real matrices and upper triangular matrices. We have the direct sum decomposition:
\begin{align}
\label{eq:kan_decomposition}
M_n\left( \C \right) & = \kfrak \oplus \afrak \oplus \nfrak
\end{align}
Any matrix $X \in M_n\left( \C \right)$ has a unique triangular decomposition into:
\begin{align}
\label{eq:triangular_decomposition}
X & = \left[ X \right]_- + \left[ X \right]_0 + \left[ X \right]_+
\end{align}
where $\left[X\right]_-$ (resp. $\left[X\right]_+) $ are respectively lower triangular and upper triangular. $\left[X\right]_0$ is diagonal. For each subspace $E$ in the direct sum decomposition \eqref{eq:kan_decomposition}, we denote the projection onto $E$ by $\Pi_E$. These projection are given for $X \in M_n(\C)$ by the expressions:
$$ \Pi_\kfrak(X) = \left[X\right]_- - \left[X\right]_-^* + i \Im \left[X\right]_0
   ; \quad
   \Pi_\afrak(X) = \Re\left[X\right]_0
   ; \quad
   \Pi_\nfrak(X) = \left[X + X^*\right]_+ \ .
$$
Indeed, one can easily check that $\Pi_\kfrak + \Pi_\afrak + \Pi_\nfrak = \id_{M_n(\C)}$.

\subsection{Flow definition}

Let $\Tc$ be the space of symmetric tridiagonal matrix form:
$$ X = \begin{pmatrix}
        a_1   &   b_1 &     0 & \dots & 0       & 0\\
        b_1   &   a_2 &   b_2 & \dots & 0       & 0\\
        \dots & \dots & \dots & \dots & \dots   & \dots\\
          0   &     0 &     0 & \dots & a_{n-1} & b_{n-1}\\
          0   &     0 &     0 & \dots & b_{n-1} & a_n\\
       \end{pmatrix},
$$
with $b_j > 0$.

The Toda flow acts on $\Tc$ via the differential equation:
\begin{align}
\label{eq:qr_flow}
\dot{X} = \left[ X, \Pi_\kfrak(X) \right] \ . 
\end{align}

Formally, define the vector field $V$ at $x \in \Tc$ by:
$$ V_x = [ x, \Pi_\kfrak(x)] \ .$$
And the Toda flow is obtained by exponentiating the vector field $V$. In fact the flow is equivalent to the pair of equations:
$$ \left\{
   \begin{array}{cc}
   X_t       & = Q_t \Lambda Q_t^*\\
   \dot{Q}   & = -\Pi_\kfrak\left( X \right) Q
   \end{array}
   \right.
$$
by noticing that equation \eqref{eq:qr_flow} can be rewritten:
$$ \dot{X} = \left[ \dot{Q} Q^{-1}, X \right]$$

Thanks to this Lax-pair formulation, we see that the flow acts by isospectral transformations. 

\begin{rmk}[The Toda flow preserves $\Tc$]
\label{rmk:QR_stabilizers}
Positivity of the extra-diagonal is preserved because equation \eqref{eq:qr_flow} implies:
 \begin{align}
\label{eq:b_i_dot}
\dot{b_i} & = b_i \left( a_i - a_{i+1} \right) 
\end{align}
\end{rmk}

\subsection{Long time behavior: the sorting property}
We assume that $X_0$ is diagonalizable with distinct eigenvalues $\left( \Lambda_1, \dots, \Lambda_n \right)$ with 
$$ \Lambda_1 > \Lambda_2 > \dots > \Lambda_n \ .$$
The diagonalized form of the initial data is given by:
$$ X_0 = Q_0 \Lambda Q_0^* \ .$$
It is known in numerical analysis that flow performs continuously the QR algorithm and thus gives the spectrum in long time. Similar results hold for other families of isospectral transformations (see the survey \cite{bib:watkins84}).

\begin{thm}[Symes \cite{bib:Symes} in tridiagonal case, \cite{bib:Chu84} in complex case]
\label{thm:toda_diagonalises}
The Toda flow diagonalizes matrices in the following sense:
\begin{itemize}
 \item The Toda flow interpolates in continuous time the Arnoldi-Lanczos-QR algorithm used in numerical analysis.
 \item We have the following long time behavior: 
$$ \lim_{t \rightarrow  \infty} X_t = \Lambda$$
$$ \lim_{t \rightarrow -\infty} X_t = w_0 \Lambda  w_0$$
where $w_0$ is the permutation matrix reversing the order of the canonical basis. In the context of reflection groups, $w_0$ is the longest element in the symmetric group $\mathfrak{S}_n$ when written as a product of transpositions $(i \ i+1)$.
\end{itemize}
\end{thm}
From this result arises the idea of keeping track of the distribution of the matrix $T_\beta$ in Eq. \eqref{eq:T_beta}, continuously throughout the flow.

\begin{rmk}
\label{rmk:u_angle_coordinate}
A classical remark in eigenvalue problems is that $Q$ has a special structure because it conjugates a diagonal matrix to tridiagonal matrix. In fact, $Q_t$ can be entirely recovered from the first row of the matrix. See Theorem 7.2.1 in Partlett \cite{bib:Parlett}. This first row $u$ is important because it plays the role of angle coordinates in the integrable Toda flow.
\end{rmk}

\subsection{Flow on tridiagonal matrices}
Let us now explain how the specialisation to tridiagonal matrices gives the original Toda Hamiltonian flow. Let $n$ identical particles, seen as point masses with mass normalized to $1$. The configuration space is then $\R^{2n}$. A configuration is a pair $\left( p, q \right) \in \R^{2n}$ where $p$ are momenta and $q$ are positions.

The dynamical system defined by Toda \cite{bib:todabook} is the Hamiltonian system associated to
\begin{align}
\label{def:toda}
   H = & \frac{\|p\|^2}{2} + V(q) \ ,
\end{align}
with $V$ the Toda potential:
$$ V(q) = \sum_{i=1}^{n-1} 2 e^{-(q_i - q_{i+1})} \ .$$

Therefore, the equations of motion are given by:
\begin{align}
\label{eq:q_dynamic}
\dot{q}_j = & \frac{\partial H}{\partial p_j} = p_j \ ,\\
\label{eq:p_dynamic}
\dot{p}_j = & -\frac{\partial H}{\partial q_j} = 2 e^{-(q_j - q_{j+1})} - 2 e^{-(q_{j-1} - q_{j})} \ .
\end{align}
The previous equations are valid for all indices $1 \leq j \leq n$ by by considering $q_0 = -\infty$ and $q_{n+1} = \infty$.

Because the center of mass has a uniform dynamic, we can assume that it is fixed and reduce the configuration space to $\R^{2(n-1)}$. Now, if one introduces the Flaschka variables
$$ a_i = p_i$$
$$ b_i = 2 e^{-\half (q_i - q_{i+1})}$$
and forms the tridiagonal matrix
$$ X = \begin{pmatrix}
        a_1   &   b_1 &     0 & \dots & 0       & 0\\
        b_1   &   a_2 &   b_2 & \dots & 0       & 0\\
        \dots & \dots & \dots & \dots & \dots   & \dots\\
          0   &     0 &     0 & \dots & a_{n-1} & b_{n-1}\\
          0   &     0 &     0 & \dots & b_{n-1} & a_n\\
       \end{pmatrix},
$$
then the Hamiltonian dynamic \eqref{eq:q_dynamic} \eqref{eq:p_dynamic} is exactly equivalent to the Toda flow \ref{eq:qr_flow} acting on real tridiagonal matrices. Note that in the Flaschka variables, the Hamiltonian takes the form
$$
H = \half \sum_{i=1}^{n}a_i^2+\sum_{i=1}^{n}b_i^2
  = \half \Tr X^2
\ .
$$

From the diagonalisation property in Theorem \ref{thm:toda_diagonalises}, we easily deduce the crude scattering behavior as $t \rightarrow \infty$:
\begin{align}
  \label{eq:crude_scattering1}
  p_i(t) = \Lambda_i + o(1)       \quad , \quad q_i(t) = \Lambda_i t + o(t) \ ,
\end{align}
and as $t \rightarrow -\infty$:
\begin{align}
  \label{eq:crude_scattering2}
  p_i(t) = \Lambda_{n-i+1} + o(1) \quad , \quad q_i(t) = \Lambda_{n-i+1} t + o(t) \ .
\end{align}
If one is interested in the scattering map for momenta i.e the relation between the behaviors as $t \rightarrow \pm \infty$, it is given by reordering eigenvalues in opposite order. This tantamounts to the multiplication by the permutation matrix $w_0$. The result holds in fact more generally from the works of Goodman and Wallach for Toda lattices in other Lie types \cite[Subsection 2.3]{bib:GWII}.

\section{Moser's scattering}
\label{subsection:Toda_Moser_scattering}

For our purposes, we are much more interested in the scattering of positions. In a really beautiful paper \cite{bib:Moser75}, Moser refines the $o(t)$ error in Eq. \eqref{eq:crude_scattering1} and Eq. \eqref{eq:crude_scattering2}. He finds there exists a $\delta>0$, depending on eigenvalue gaps, such that $t \rightarrow \infty$:
$$ q_i(t) = \Lambda_i t + \beta_i^+ + o(e^{-\delta |t|})$$
and as $t \rightarrow -\infty$:
$$ q_i(t) = \Lambda_{n-i+1} t + \beta_i^- + o(e^{-\delta |t|})$$
and the differences $\beta_{n-i+1}^+-\beta_i^-$ are related to a logarithmic interaction potential. The exact expression is given in \cite{bib:Moser75} eq. (4.3) and (4.4):
$$ \beta_{n-i+1}^+-\beta_i^- = 2 \sum_{j<i} \log\left| \Lambda_i - \Lambda_j \right| - 2 \sum_{i<j} \log\left| \Lambda_i - \Lambda_j \right|$$
Hence, the second order of the scattering in position variables reveals the logarithmic interaction potential for eigenvalues. That was the starting point of our investigations.

\medskip

{\bf Solution by inverse scattering.}
In order to completely solve the Toda flow, at the theoretical level, one starts by computing the diagonalization of the initial data $X_0$. Since $\Lambda$ is given by the infinite time behavior of the system, it is called the scattering data. Inverse scattering consists in using $\Lambda$ in order to compute the finite time solution, which is given by the computation of $Q$. 

Nevertheless, due to the special structure of $Q$ discussed in Remark \ref{rmk:u_angle_coordinate}, there is no need to compute the entire matrix. It suffices to keep track of the vector:
$$ u := Q^{-1} e$$
where $e = e_1$ is the first vector in $\C^n$. 
   
   \begin{proposition}
   \label{proposition:u_solution}
   The vector $u$ follows the dynamic:
   $$ u_t = \frac{e^{\Lambda t} u_0}{ \| e^{\Lambda t} u_0 \| }$$
   \end{proposition}
   \begin{proof}
   Since
   $$ \dot{Q} = -\Pi_\kfrak X Q
              = -\left( X - \left[X\right]_0 - \left[X\right]_-^* - \left[X\right]_+ \right) Q
              = -Q \Lambda + \left( \left[X\right]_0 + \left[X^*+X\right]_+ \right) Q \ ,$$
   we have that:
   $$ \frac{d}{dt}\left( Q_t e^{\Lambda t} \right)
    = \left( \left[X\right]_0 + \left[X^*+X\right]_+ \right) Q_t e^{\Lambda t}$$
   The previous equation is a right-invariant autonomous equation, with upper triangular increments. As a consequence, there exists an upper triangular matrix $T_t$ with positive diagonal such that:
   $$ Q_t e^{\Lambda t} = T_t Q_0 \ ,$$
   and hence:
   $$ u_t = Q_t^{-1} e = e^{\Lambda t} Q_0^{-1} T_t^{-1} e = \frac{e^{\Lambda t} u_0}{[T_t]_{11}} \ .$$
   The proof is finished as $u_t$ needs to be of norm $1$.
   \end{proof}
   
   \subsection{Precise scattering asymptotics.}
   
   As mentioned before, one sees the appearance of the Vandermonde in the second order scattering asymptotics for positions, in the real tridiagonal case, thanks to Moser's result \cite{bib:Moser75}. His approach relied on real analyticity of the flow and seems difficult to adapt or to generalize. We will rather use an orthogonal polynomial technique that expresses the action variable in a form more amenable to asymptotics. The technique is used on real tridiagonal matrices to express orthogonal polynomials thanks to the coefficients in the three term recurrence (See Chapter II in \cite{bib:Szego75}). Its application to the tridiagonal Toda has been implemented in handwritten lecture notes of Deift, the author managed to get his hands on.
   
   Let us introduce the Gram determinant, using the unit vector $e=e_1$ again:
   $$ \Delta_k
    = \det\left( \langle X^{i-1} e, X^{j-1} e \rangle \right)_{i,j=1}^k \ .
   $$
   The scattering of positions can be observed from the convergence of the $b_i$ to zero, and the analogue of Moser's scattering result comes from the asymptotic analysis of $\Delta_k$ as:
   \begin{align}
   \label{eq:delta_to_b}
   \Delta_k & = \prod_{i=1}^k b_i^{2(k-i)} \ .
   \end{align}
   
   Indeed, notice that, for all $k \leq n$:
   $$ X^k e = \left( \prod_{i=1}^k b_i \right) e_{k+1} + h_k$$
   where $h_k \in \textrm{Span}_\C\left( e_1, \dots, e_{k} \right)$. Hence:
   $$
     e \wedge X e \wedge X^2 e \wedge \dots \wedge X^{k-1} e\\
   = \left( \prod_{l=1}^{k-1} \prod_{i=1}^l b_i \right) e_1 \wedge e_2 \wedge e_3 \wedge \dots \wedge e_k
   $$
   and
   $$ \Delta_k
    = \det\left( \langle X^{i-1} e, X^{j-1} e \rangle \right)_{i,j=1}^k
    = \| e \wedge X e \wedge X^2 e \wedge \dots \wedge X^{k-1} e \|^2
    = \prod_{i=1}^k b_i^{2(k-i)}
   $$
   
   \begin{thm}[Precise scattering asymptotics]
   \label{thm:precise_scattering}
   As $t \rightarrow \infty$:
   $$ \Delta_k(t) = e^{2 \left( \sum_{l=1}^k \Lambda_1 - \Lambda_l \right) t}
                    \left| \Delta\left(\Lambda_1, \dots, \Lambda_k \right) \right|^2 
                    \prod_{l=1}^k \left| \frac{u_0(l)}{u_0(1)} \right|^2 (1 + o(1))
   $$ 
   In particular, the initial angle coordinates can be read from asymptotics:
   $$ \left| \frac{u_0(l)}{u_0(1)} \right|
               \sim
               e^{\left( \Lambda_1 - \Lambda_k \right) t}
               \frac{ \left| \Delta\left(\Lambda_1, \dots, \Lambda_{k-1} \right) \right| }
                    { \left| \Delta\left(\Lambda_1, \dots, \Lambda_{k  } \right) \right| }
               \sqrt{ \frac{ \Delta_{k}(t) }{ \Delta_{k-1}(t)} }
   $$
   \end{thm}
   
   This is nothing but Moser's result with a different proof. Our proof (actually Deift's) is proved by introducing the probability measure on $\R$ - often called the "spectral measure" in the literature devoted to Jacobi operators:
   $$ \mu_t\left(d\lambda \right)
      =
      \sum_{k=1}^n \left| \langle e_k, u_t \rangle \right|^2 \delta_{\Lambda_k}( d\lambda ),$$
   thanks to which the Gram determinant has a nice formula:
   \begin{proposition}
   \label{proposition:delta_determinantal_formula}
   $$ \Delta_k(t)
    = \int_{\lambda_1 > \lambda_2 > \dots > \lambda_k }
      \left| \Delta\left(\lambda_1, \dots, \lambda_k \right) \right|^2 \prod_{l=1}^k \mu_t\left(d\lambda_l \right)$$ 
   In particular, we have the exact result:
   $$ \Delta_n(t)
    = \left| \Delta\left(\Lambda_1, \dots, \Lambda_n \right) \right|^2 
      \prod_{l=1}^n |\langle e_l, u_t \rangle|^2 \ .
   $$ 
   \end{proposition}
   \begin{proof}
   We have:
   \begin{align*}
        \langle X_t^{i-1} e, X_t^{j-1} e \rangle
    = & \langle Q_t \Lambda^{i-1} Q_t^* e, Q_t \Lambda^{j-1} Q_t^* e \rangle \\
    = & \langle \Lambda^{i-1} u_t, \Lambda^{j-1} u_t \rangle \\
    = & \sum_{k} \Lambda_k^{i-1+j-1} \left| \langle e_k, u_t \rangle \right|^2 \\
    = & \int_\R \lambda^{i-1+j-1} \mu_t\left( d\lambda \right) \ .
   \end{align*}
   And therefore, by multilinearity of the determinant with respect to columns:
   \begin{align*}
       \Delta_k(t)
   = & \det\left( \langle X^{i-1}_t e, X^{j-1}_t e \rangle \right)_{i,j=1}^k \\
   = & \det\left( \int_\R \lambda_j^{i-1+j-1} \mu_t\left( d\lambda_j \right) \right)_{i,j=1}^k \\
   = & \int_{\R^k} \det\left( \lambda_j^{i-1} \lambda_j^{j-1} \right)_{i,j=1}^k \prod_{l=1}^k \mu_t\left(d\lambda_l \right)\\ 
   = & \int_{\R^k} \det\left( \lambda_j^{i-1} \right)_{i,j=1}^k \prod_{j=1}^k \lambda_j^{j-1} \prod_{l=1}^k \mu_t\left(d\lambda_l \right)\\ 
   = & \int_{\R^k} \Delta( \lambda_1, \dots, \lambda_k ) \prod_{l=1}^k \lambda_l^{l-1} \prod_{l=1}^k \mu_t\left(d\lambda_l \right) \ .
   \end{align*}
   Using a standard anti-symmetrization trick:
   \begin{align*}
       \Delta_k(t)
   = & \frac{1}{k!} \sum_{\sigma \in W} \int_{\R^k} \varepsilon(\sigma) 
       \Delta\left(\lambda_1, \dots, \lambda_k \right)
       \prod_{l=1}^k \lambda_{\sigma(l)}^{l-1}
       \prod_{l=1}^k \mu_t\left(d\lambda_l \right)\\ 
   = & \frac{1}{k!} \int_{\R^k} \left| \Delta\left(\lambda_1, \dots, \lambda_k \right) \right|^2 \prod_{l=1}^k \mu_t\left(d\lambda_l \right) \ .
   \end{align*}
   And by symmetry, we can order the integration variables:
   $$ \Delta_k(t)
    = \int_{\lambda_1 > \lambda_2 > \dots > \lambda_k }
      \left| \Delta\left(\lambda_1, \dots, \lambda_k \right) \right|^2 \prod_{l=1}^k \mu_t\left(d\lambda_l \right) \ .$$
   \end{proof}
   
   Moreover, $\mu_t$ has a simple dynamic:
   \begin{proposition}
   \label{proposition:mu_dynamic}
   $$ \mu_t\left(d\lambda \right)
    = \frac{e^{2 \lambda t}}{\left\| e^{\Lambda t} u_0 \right\|^2}
      \mu_0(d\lambda) 
   $$ 
   \end{proposition}
   \begin{proof}
   From the the dynamic of $u$ in Proposition \ref{proposition:u_solution}, we obtain:
   \begin{align*}
       \mu_t\left(d\lambda \right)
   = & \sum_{k} \left| \langle e_k, u_t \rangle \right|^2\delta_{\Lambda_k}(d\lambda) \\
   = & \frac{1}{\left\| e^{\Lambda t} u_0 \right\|^2}
       \sum_{k} e^{2 \Lambda_k t} \left| \langle e_k, u_0 \rangle \right|^2\delta_{\Lambda_k}(d\lambda) \\
   = & \frac{1}{\left\| e^{\Lambda t} u_0 \right\|^2}
       e^{2 \lambda t} \mu_0(d\lambda)  
   \end{align*}
   \end{proof}
   
   The Gram determinant has a nice formula in terms of the previous measure:

   \begin{proof}[Proof of theorem \ref{thm:precise_scattering}]
   The combination of the two previous propositions gives:
   $$ \Delta_k(t) = \| e^{\Lambda t} u_0 \|^{-2k}
                 \int_{\lambda_1 > \lambda_2 > \dots > \lambda_k }
                 e^{2 \left( \sum_{l=1}^k \lambda_l \right) t}
                 \left| \Delta\left(\lambda_1, \dots, \lambda_k \right) \right|^2 \prod_{l=1}^k 
                 \mu_0\left( d\lambda_l \right)$$
   As $t \rightarrow \infty$, the integral's dominant terms are obtained by picking up only the $k$ largest eigenvalues $\Lambda_1 > \Lambda_2 > \dots > \Lambda_k$. Hence the asymptotics:
   \begin{align*}
   \Delta_k(t) = &
                 \| e^{\Lambda t} u_0 \|^{-2k}
                 \int_{\lambda_1 > \lambda_2 > \dots > \lambda_k }
                 e^{2 \left( \sum_{l=1}^k \lambda_l \right) t}
                 \left| \Delta\left(\lambda_1, \dots, \lambda_k \right) \right|^2 \prod_{l=1}^k 
                 \mu_0\left( d\lambda_l \right)\\
   = &  \| e^{\Lambda t} u_0 \|^{-2k}
        e^{2 \left( \sum_{l=1}^k \Lambda_l \right) t} \left| \Delta\left(\Lambda_1, \dots, \Lambda_k \right) \right|^2 
       \prod_{l=1}^k \left| \langle e_l, u_0 \rangle \right|^2 (1 + o(1))
   \end{align*}
   Combining that fact with the asymptotics for $\| e^{\Lambda t} u_0 \|$ yields the result.
   \end{proof}
   In particular, Moser's result appears explicitly in:
   \begin{align*}
      & b_k^{2}(t)\\
    = & \frac{ \Delta_{k} \Delta_{k-2}}{\Delta_{k-1}^2}\\
    = & e^{2 \left( \Lambda_k - \Lambda_{k-1} \right) t}
        \frac{ \left| \Delta\left(\Lambda_1, \dots, \Lambda_k \right) \right|^2 \left| \Delta\left(\Lambda_1, \dots, \Lambda_{k-1} \right) \right|^2 }
             { \left| \Delta\left(\Lambda_1, \dots, \Lambda_{k-2} \right) \right|^4 }
        \frac{ \left| u_0(k) \right|^2 \left| u_0(k-1) \right|^2}{ \left| u_0(k-2) \right|^4 }
        (1 + o(1) )
   \end{align*}
               
   \subsection{Scattering map}
   Consider the vector field on $\Tc$:
   $$ \forall x \in \Tc, \ V^{free}_x := -\sum_{i=1}^{n-1} E_{i+1, i} b_i \left( a_i - a_{i+1}\right)$$
   whose flow gives:
   $$ e^{t V^{free}} \cdot b_i = b_i e^{-\left( a_i - a_{i+1}\right)t}$$
   
   In the scattering regime, diagonals are constant and the dynamic of position is linear in time. Therefore, the ``free'' dynamic corresponding to isolated particles is given by the flow $e^{t V^{free}}$. The question at hand is the behavior of the sequence of diffeomorphisms:
   $$ e^{-t V} \cdot e^{t V^{free}} $$
   as $t \rightarrow \infty$. Clearly, the flow generated by $V^{free}$ is not isospectral, but in the regime where the extra-diagonals are small, it almost is.
   
   \begin{thm}
   \label{thm:scattering_map}
   Let $\Delta$ be the set of matrices with increasing diagonal entries. There exists a map $\Sc: \Tc \cap \Delta \rightarrow \Tc$ such that:
   $$ \Sc = \lim_{t \rightarrow \infty} e^{-t V} \cdot e^{t V^{free}}$$ 
   We have:
   $$ \Sc x = Q [x]_{0} Q^*$$
   where the first row of $Q$, $u_0$ satisfies:
   $$ \langle e_k, u_0 \rangle = \frac{ \prod_{i=1}^{k-1} b_i(0) }{ \prod_{i=1}^{k-1} \left| \Lambda_k - \Lambda_i \right| } \langle e_1, u_0 \rangle \ .$$
   \end{thm}
   \begin{proof}
   Let $x$ be a matrix in $\Tc \cap \Delta$. Then:
   $$ e^{t V^{free}} \cdot x = [x]_{0} + \sum_{i} x_{i+1,i} e^{-\left( x_{i,i} - x_{i+1,i+1} \right)t} \left( E_{i+1, i} + E_{i, i+1} \right)$$
   Therefore, $e^{t V^{free}} \cdot x$ converges exponentially fast to a diagonal matrix, for which the sorted eigenvalues lie on the diagonal. The conserved quantities by the flow $e^{-t V}$ are therefore exponentially close to $[x]_{0}$.
   
   Now consider an asymptotic behavior, obtained after a free evolution:
   $$ b_k(t) = b_k(0) e^{\left( \Lambda_{k+1} - \Lambda_k \right) t } \left( 1 + o(1) \right)$$
   Thanks to Theorem \ref{thm:precise_scattering}, this corresponds to a scattering state at time $t$ for the Toda evolution which had at time $0$ the angle coordinates:
   \begin{align*}
   \frac{|v_0(k)| }{|v_0(1)|}
            & = e^{\left( \Lambda_1 - \Lambda_k \right) t}
                \sqrt{ \frac{\Delta_k(t) }{ \Delta_{k-1}(t)} }
                \frac{\left| \Delta(\Lambda_1, \dots, \Lambda_{k-1}) \right|} {\left| \Delta(\Lambda_1, \dots, \Lambda_k) \right|} \left( 1 + o(1) \right)\\
            & = e^{\left( \Lambda_1 - \Lambda_k \right) t}
                \left( \prod_{i=1}^{k-1} b_i(t) \right)
                \frac{\left| \Delta(\Lambda_1, \dots, \Lambda_{k-1}) \right|} {\left| \Delta(\Lambda_1, \dots, \Lambda_k) \right|}
                                                                        (1 + o(1)) \\
            & = \frac{ \prod_{i=1}^{k-1} b_i(0) }{ \prod_{i=1}^{k-1} \left| \Lambda_k - \Lambda_i \right| } (1 + o(1))\\
   \end{align*}
   Therefore $e^{-t V} \cdot e^{t V^{free}} \cdot x$ converges to the element with action variables $[x]_{0}$ and angle coordinates $v_0$. 
   \end{proof}
   
   From the previous proof and from Proposition \ref{proposition:delta_determinantal_formula}, we notice that:
   \begin{align}
   \label{eq:delta_of_scattering}
   \Delta_n\left( \Sc x \right) & = \left| \Delta\left(\Lambda_1, \dots, \Lambda_n \right) \right|^2 
                                    \prod_{k=1}^n |v_0(k)|^{2}
                                  = |v_0(1)|^{2n} \prod_{i=1}^{n-1} b_i^{2(n-i)}
   \end{align}

   \subsection{Invariant differential forms}
   
   Define the volume form on $\Tc$:
   \begin{align}
   \label{eq:form_T}
   \omega_\Tc := \wedge_{i=1}^{n-1} \frac{db_i}{b_i} \wedge \omega_\afrak 
   \end{align}
   where
   $$ \omega_{\afrak}  = \wedge_{i=1}^n d a_i \ .$$
   As we will use these forms as reference integration measures, it is crucial that are invariant under the flow.
   
   \begin{proposition}[Form invariance]
   The Toda flow preserves the volume form $\omega_\Tc$:
   $$ \Lc_{V} \omega_{\Tc} = 0$$
   In particular:
   $$ \left( e^{tV} \right)_* \omega_\Tc = 0$$
   \end{proposition}
   \begin{proof}
   In the positions and momenta coordinates of the classical Toda Hamiltonian, the form $i_* \omega_{\Tc}$ can be written as:
   $$ i_* \omega_{\Tc} = \left( \wedge_{i=1}^{n} dp_i \right) \wedge \left( \wedge_{i=1}^{n-1} dq_i - dq_{i+1} \right)$$
   which is the Liouville form of the Toda system. In this case, not only $i_* \omega_{\Tc}$ is preserved but there is an underlying symplectic form that is preserved. 
   \end{proof}

   \section{Novel proof of Theorem \ref{thm:DumitriuEdelman} via scattering}
   \label{section:Toda_proof_main}
   
   Consider a random matrix $X_0 \in \Tc$ distributed as in \eqref{eq:T_beta}. We are interested in the distribution of eigenvalues $\Lambda$ as well as the distribution of the vector $u_0$. Clearly, from the properties of the Toda flow, for every bounded continuous function:
   $$ \E\left( f( \Lambda, u_0) \right)
    = \lim_{t \rightarrow \infty} \E\left( f\left( [X_t]_{0}, u_0 \right) \right)$$
   Now, from the explicit expression of $\chi$ and Gaussian distributions, we have:
   \begin{align*}
     & \E\left( f\left( [X_t]_{0}, u_0 \right)\right)\\
   = & C_{n, \beta} \int_\Tc f\left( [e^{tV} \cdot X_0]_{0}, v_0\left( X_0 \right) \right) \prod_{i=1}^{n-1} b_i(X_0)^{2 \beta (n-i) } e^{-\half \| X_0 \|^2 } \omega_{\Tc}(dX_0)
   \end{align*}
   where 
   $$ C_{n, \beta} = \frac{1}{(2\pi)^{n(n+1)} \prod_{j=1}^{n-1} \Gamma( \half j \beta)} \ .$$
   Thanks to Equation \eqref{eq:delta_to_b}:
   \begin{align*}
     & \E\left( f\left( [X_t]_{0}, u_0 \right)\right)\\
   = & C_{n, \beta} \int_\Tc f\left( [e^{tV} \cdot X_0]_{0}, u_0\left( X_0 \right) \right) \Delta_n( X_0 )^{\beta} e^{-\half \| X_0 \|^2 } \omega_{\Tc}(dX_0)
   \end{align*}
   Then, we make successively the change of variable by the flow $e^{tV}$ and $e^{t V^{free}}$. Thanks to the invariance of measures:
   \begin{align*}
     & \E\left( f\left( [X_t]_{0}, u_0 \right)\right)\\
   = & C_{n, \beta} \int_\Tc f\left( [e^{tV} \cdot X_0]_{0}, u_0\left( e^{-t V} \cdot e^{tV} \cdot X_0 \right) \right) \\
     & \quad \quad \quad \Delta_n(  e^{-t V} \cdot e^{tV} \cdot X_0 )^{\beta} e^{-\half \| X_0 \|^2 } \omega_{\Tc}(dX_0)\\
   = & C_{n, \beta} \int_\Tc f\left( [X_0]_{0}, u_0\left( e^{-t V} \cdot X_0 \right) \right) \\
     & \quad \quad \quad \Delta_n(  e^{-t V} \cdot X_0 )^{\beta} e^{-\half \| X_0 \|^2 } \omega_{\Tc}(dX_0)\\
   = & C_{n, \beta} \int_\Tc f\left( [X_0]_{0}, u_0\left( e^{-t V} \cdot e^{t V^{free}} \cdot X_0 \right) \right)\\
     & \quad \quad \Delta_n(  e^{-t V} \cdot e^{t V^{free}} \cdot X_0 )^{\beta}
                e^{-\half \| e^{t V^{free}} \cdot X_0 \|^2 } \omega_{\Tc}(dX_0) \ .
   \end{align*}
   Now, as $t \rightarrow \infty$, we have the appearance of the scattering map $\Sc = \lim_{t \rightarrow \infty} e^{-t V} \cdot e^{t V^{free}}$. By assuming $f$ supported on $\Tc \cap \Delta$ in the first variable, there is no need to bother about the initial data $X_0$ having its eigenvalues ordered. Moreover, the coefficients of the extra-diagonal in $e^{t V^{free}} \cdot X_0$ asymptotically vanish. Hence:
   \begin{align*}
     & \E\left( f( \Lambda, u_0) \right)\\
   = & C_{n, \beta} \int_\Tc f\left( [X_0]_{0}, u_0\left( \Sc X_0 \right) \right)
                \Delta_n(  \Sc X_0 )^{\beta}
                e^{-\half \| [X_0]_{0} \|^2 } \omega_{\Tc}(dX_0)\\
   = & C_{n, \beta} \int_{\lambda_1 > \lambda_2 > \dots > \lambda_n} d \Lambda \ e^{ -\half \sum |\Lambda_i|^2 }\\
     &        \int_{\left( \R^+ \right)^{n-1} } \wedge_{i=1}^{n-1} \frac{db_i}{b_i} \ f\left( \Lambda, u_0\left( \Sc X_0 \right) \right)
                                            \Delta_n(  \Sc X_0 )^{\beta} \ .
   \end{align*}
   
   Finally, we conclude thanks to the identity \eqref{eq:delta_of_scattering} that:
   \begin{align*}
     & \E\left( f( \Lambda, u_0) \right)\\
   = & C_{n, \beta} \int_{\lambda_1 > \lambda_2 > \dots > \lambda_n} d \Lambda e^{ -\half \sum |\Lambda_i|^2} \\
     &              \int_{\left( \R^+ \right)^{n-1} } \wedge_{i=1}^{n-1} \frac{db_i}{b_i} \ \left( \prod_{i=1}^{n-1} b_i^{2\beta(n-i)} \right)
                                                                                      |u_0(1)|^{2 \beta n} f\left( \Lambda, u_0\left( \Sc X_0 \right) \right) \ .
   \end{align*}
   Moreover, by changing the integration variables to $c_k$:
   $$ \frac{|u_0(k+1)|}{|u_0(1)| }
    = \frac{ \prod_{i=1}^{k-1} b_i }{ \prod_{i=1}^{k-1} \left| \Lambda_k - \Lambda_i \right| }
    = c_k \ , 
   $$
   and knowing that:
   $$ \prod_{i=1}^{n-1} b_i^{(n-i)} = \left| \Delta\left(\Lambda_1, \dots, \Lambda_n \right) \right| \prod_{i=1}^{n-1} c_i \ ,$$
   we obtain:
   \begin{align*}
     & \E\left( f( \Lambda, u_0) \right)\\
   = & C_{n, \beta} \int_{\lambda_1 > \lambda_2 > \dots > \lambda_n} d \Lambda \left| \Delta\left(\Lambda_1, \dots, \Lambda_n \right) \right|^{2\beta} e^{ -\half \sum |\Lambda_i|^2} \\
     &  \int_{\left( \R^+ \right)^{n-1} } \wedge_{i=1}^{n-1} \frac{dc_i}{c_i} \ \left( \prod_{i=1}^{n-1} c_i^{2\beta} \right)
        |u_0(1)|^{2 \beta n} f\left( \Lambda, |u_0(1)|\left(1, c_1, c_2 \dots \right) \right) \ .
   \end{align*}
   We recognize the distribution of the Gaussian $\beta$-ensemble for the spectrum $\Lambda$. Moreover, the distribution of the vector $u_0$ is given by the $\beta$-Dirichlet distribution. 
   
   This concludes the proof of Theorem \ref{thm:DumitriuEdelman}.

\bibliographystyle{alpha}
\bibliography{HDR.bib}

\begin{thebibliography}{Wat84}

\bibitem[Chu84]{bib:Chu84}
Moody~T. Chu.
\newblock The generalized {T}oda flow, the {${\rm QR}$} algorithm and the
  center manifold theory.
\newblock {\em SIAM J. Algebraic Discrete Methods}, 5(2):187--201, 1984.

\bibitem[DE02]{bib:DumitriuEdelman}
Ioana Dumitriu and Alan Edelman.
\newblock Matrix models for beta ensembles.
\newblock {\em J. Math. Phys.}, 43(11):5830--5847, 2002.

\bibitem[GW84]{bib:GWII}
Roe Goodman and Nolan~R. Wallach.
\newblock Classical and quantum mechanical systems of {T}oda-lattice type.
  {II}. {S}olutions of the classical flows.
\newblock {\em Comm. Math. Phys.}, 94(2):177--217, 1984.

\bibitem[Mos75]{bib:Moser75}
J{\"u}rgen Moser.
\newblock Finitely many mass points on the line under the influence of an
  exponential potential--an integrable system.
\newblock In {\em Dynamical systems, theory and applications ({R}encontres,
  {B}attelle {R}es. {I}nst., {S}eattle, {W}ash., 1974)}, pages 467--497.
  Lecture Notes in Phys., Vol. 38. Springer, Berlin, 1975.

\bibitem[OV96]{olshanski1996ergodic}
Grigori Olshanski and Anatoli Vershik.
\newblock Ergodic unitarily invariant measures on the space of infinite
  hermitian matrices.
\newblock {\em arXiv preprint math/9601215}, 1996.

\bibitem[Par98]{bib:Parlett}
Beresford~N. Parlett.
\newblock {\em The symmetric eigenvalue problem}, volume~20 of {\em Classics in
  Applied Mathematics}.
\newblock Society for Industrial and Applied Mathematics (SIAM), Philadelphia,
  PA, 1998.
\newblock Corrected reprint of the 1980 original.

\bibitem[Sym82]{bib:Symes}
W.~W. Symes.
\newblock The {$QR$} algorithm and scattering for the finite nonperiodic {T}oda
  lattice.
\newblock {\em Phys. D}, 4(2):275--280, 1981/82.

\bibitem[Sze75]{bib:Szego75}
G{\'a}bor Szeg{\"{o}}.
\newblock {\em Orthogonal polynomials}.
\newblock American Mathematical Society, Providence, R.I., fourth edition,
  1975.
\newblock American Mathematical Society, Colloquium Publications, Vol. XXIII.

\bibitem[Tod89]{bib:todabook}
Morikazu Toda.
\newblock {\em Theory of nonlinear lattices}, volume~20 of {\em Springer Series
  in Solid-State Sciences}.
\newblock Springer-Verlag, Berlin, second edition, 1989.

\bibitem[Wat84]{bib:watkins84}
David~S. Watkins.
\newblock Isospectral flows.
\newblock {\em SIAM Rev.}, 26(3):379--391, 1984.

\end{thebibliography}

\end{document}